\documentclass[11pt]{article}
\usepackage[margin=1in]{geometry}

\usepackage[table,xcdraw,dvipsnames]{xcolor}
\usepackage{graphicx,caption,subcaption,dsfont,changepage,enumitem,algorithm,algpseudocode,microtype,xargs,tikz,amsmath,amsthm,amssymb,amstext,mathrsfs,mathtools,bm,ytableau}
\usepackage[most]{tcolorbox}

\usepackage[colorlinks=true]{hyperref}
\usepackage[capitalize,nameinlink,noabbrev]{cleveref}

\hypersetup{
    colorlinks=true,
    linkcolor=NavyBlue,
    filecolor=magenta,
    urlcolor=teal,
    pdfpagemode=FullScreen,
    citecolor=OliveGreen
}

\def\BibTeX{{\rm B\kern-.05em{\sc i\kern-.025em b}\kern-.08em
    T\kern-.1667em\lower.7ex\hbox{E}\kern-.125emX}}

\newcommand{\msf}[1]{\mathsf{#1}}

\newcommand{\ket}[1]{| #1 \rangle}

\newcommand{\bigo}{\mathcal{O}}

\newcommand{\UQCPCP}{\msf{UniqueQCPCP}}

\newcommand{\UMTP}{\msf{UniqueMTP}}

\newcommand{\SAT}{\msf{SAT}}
\newcommand{\MTP}{\msf{MTP}}
\newcommand{\BQ}{\msf{BQ}}

\newcommand{\QPH}{\msf{QPH}}

\newtheorem{theorem}{Theorem}[section]
\newtheorem{lemma}[theorem]{Lemma}

\newtheorem{corollary}[theorem]{Corollary}

\newtheorem{definition}[theorem]{Definition}

\newtheorem{remark}[theorem]{Remark}

\numberwithin{equation}{section}

\begin{document}

\title{Collapses in quantum-classical probabilistically checkable proofs \\ and the quantum polynomial hierarchy}

\author{Kartik Anand\thanks{Indian Institute of Technology Goa, \href{mailto:kartik.anand.19031@iitgoa.ac.in}{kartik.anand.19031@iitgoa.ac.in}} 
\and Kabgyun Jeong\thanks{Research Institute of Mathematics, Seoul National University, \href{mailto:kgjeong6@snu.ac.kr}{kgjeong6@snu.ac.kr}}
\and Junseo Lee\thanks{Team QST, Seoul National University, \href{mailto:harris.junseo@gmail.com}{harris.junseo@gmail.com} (Current address: Norma Inc.)}
}

\footnotetext[1]{Authors contributed equally and are listed alphabetically by last name.}

\date{\today}

\maketitle

\begin{abstract}
We investigate structural properties of quantum proof systems by establishing several collapse results that uncover simplifications in their complexity landscape. By extending the classical results such as the Karp–Lipton theorem to quantum polynomial hierarchy with quantum proofs and establishing uniqueness preservation for quantum-classical probabilistically checkable proof systems, we address how different constraints affect the computational power of these proof systems.

Our main contributions are:
\begin{itemize}
\item We prove that restricting quantum-classical PCP systems to uniqueness does not reduce computational power: $\mathsf{UniqueQCPCP} = \mathsf{QCPCP}$ under $\mathsf{BQ}$-operator and randomized reductions. This equivalence demonstrates robustness, as is shown by $\mathsf{UniqueQCMA = QCMA}$ under randomized reductions result, showing that uniqueness constraints do not fundamentally limit the capabilities even in these $\mathsf{PCP}$-type proof systems.

\item We establish a non-uniform quantum analogue of the Karp–Lipton theorem, concluding a conditional ``inclusive collapse'' for the quantum polynomial hierarchy: if $\mathsf{QMA} \subseteq \mathsf{BQP}/\mathsf{qpoly}$, then $\mathsf{QPH} \subseteq \mathsf{Q\Sigma}_2/\mathsf{qpoly}$. This extends the classical collapse theorem to quantum complexity with quantum advice, providing evidence against efficient quantum advice for $\mathsf{QMA}$-complete problems.

\item We introduce a bounded-entanglement variant of the quantum polynomial hierarchy, denoted by $\mathsf{BEQPH}$, which imposes entanglement constraints across interaction rounds, analogous to the unbounded variant $\mathsf{QEPH}$. We prove that, like $\mathsf{QEPH}$, this class remains tractable due to the convex structure of each level in the hierarchy. Furthermore, we show that all levels above the fourth collapse unconditionally to the fourth level.
In a similar vein, we define the separable hierarchy $\mathsf{SepQPH}$ by setting the entanglement bounds to zero. The results established for $\mathsf{BEQPH}$ extend naturally to $\mathsf{SepQPH}$.
We emphasize that this stands in contrast to prior work on the quantum entangled polynomial hierarchy, which attributes the collapse entirely to the presence of entanglement. In our case, we observe that computational simplicity also arises from the structural nature of the protocol itself, which renders the feasible region of the associated optimization problem convex—unlike in the case of $\mathsf{QPH}$.

\end{itemize}
These results contribute to our understanding of structural boundaries in quantum complexity theory and the interplay between different constraint types in quantum proof systems.
\end{abstract}

\newpage

\tableofcontents

\newpage
\section{Introduction}
Computational complexity theory seeks to understand the fundamental limits and capabilities of computational models under various resource constraints. A particularly rich and fruitful direction within this field involves studying how classical computational machines can be augmented with different forms of auxiliary information to enhance their problem-solving capabilities. This auxiliary information may take various forms: deterministic advice strings that depend only on the input length, non-deterministic certificates or witnesses that can be verified efficiently, or interactive proofs where a computationally unbounded prover attempts to convince a computationally unbounded verifier of the truth of a statement through a sequence of message exchanges.

Over the past few decades, complexity theorists have made remarkable progress in characterizing the computational power of such augmented systems. The overall theme of this research program can be condensed in the following fundamental question:
\begin{quote}
\textit{``Given a particular verifier, how powerful can the system become when provided with additional information encoded as proofs or advice, and how do restrictions on these forms of information affect the verifier's power?''}
\end{quote}
This question has led to several landmark results that have shaped our understanding of the complexity of computational models. Perhaps one of the most surprising results is the complete characterization of interactive proof systems. Goldwasser, Micali, and Rackoff initially introduced the concept of interactive proofs~\cite{GMR1989}, demonstrating that interaction between a prover and verifier could yield computational advantages beyond traditional proof systems. This work was extended by Babai~\cite{Babai1985} with the introduction of Arthur-Merlin games, providing an alternative perspective on interactive computation. The pinnacle of this line of research came with Shamir's groundbreaking result~\cite{Shamir1992IPPSPACE}, which established that the class of problems admitting interactive proofs with three or more rounds of interaction is exactly as powerful as $\mathsf{PSPACE}$. Remarkably, this result also showed that increasing the number of interaction rounds beyond three does not enhance the computational power of the system—a surprising collapse that demonstrates the efficiency of interactive protocols. The story takes a dramatic turn when we consider quantum entanglement in multi-prover interactive proofs. While classical multi-prover interactive proofs are bounded by $\mathsf{NEXP}$, the quantum variant exhibits fundamentally different behavior. In a recent breakthrough, Ji, Natarajan, Vidick, Wright, and Yuen~\cite{10.1145/3485628} proved that $\mathsf{MIP^*} = \mathsf{RE}$, showing that multi-prover interactive proofs with quantum entangled provers can decide any recursively enumerable language. This result is particularly striking because it demonstrates that quantum entanglement allows the proof system to transcend the polynomial hierarchy entirely, reaching the computability-theoretic boundary of decidable problems. The contrast between the classical collapse to $\mathsf{PSPACE}$ and the quantum expansion to $\mathsf{RE}$ illustrates the profound impact that quantum resources can have on interactive proof systems.

Another cornerstone result that illustrates the delicate relationship between computational classes and auxiliary information is the Karp–Lipton theorem~\cite{KarpLipton1980}. This theorem establishes a conditional collapse of the polynomial hierarchy under the assumption that $\mathsf{NP} \subseteq \mathsf{P}/\mathsf{poly}$—that is, if every problem in $\mathsf{NP}$ can be solved in polynomial time with polynomial-sized advice strings. The theorem demonstrates that such an inclusion would imply $\mathsf{PH} = \mathsf{\Sigma}_2$, collapsing the entire polynomial hierarchy to its second level. Since most complexity theorists believe the polynomial hierarchy is infinite (and thus does not collapse), the Karp–Lipton theorem provides strong evidence against the existence of polynomial-sized circuits for $\mathsf{NP}$-complete problems.

The 1990s witnessed another revolutionary development with the proof of the PCP (Probabilistically Checkable Proof) theorem by Arora, Lund, Motwani, Sudan, and Szegedy~\cite{ALMSS1998}, building on earlier work by Babai, Fortnow, and Lund~\cite{BFL1991} and Feige, Goldwasser, Lovász, Safra, and Szegedy~\cite{FGLSS1996}. This celebrated result asserts that every problem in $\mathsf{NP}$ admits a probabilistically checkable proof system with remarkable efficiency properties: specifically, any $\mathsf{NP}$ statement has a proof that can be verified by reading only a constant number of randomly chosen bits from the proof string, using only $\mathcal{O}(\log (n))$ random bits in total. The PCP theorem not only provided profound insights into the structure of $\mathsf{NP}$ but also established fundamental connections between computational complexity and approximation algorithms, showing that many optimization problems cannot be approximated within certain factors unless $\mathsf{P} = \mathsf{NP}$.

With the advent of quantum computing, researchers have naturally sought to understand quantum analogues of these classical results. The quantum setting introduces fundamentally new phenomena—such as superposition, entanglement, and measurement—that can potentially provide computational advantages while also introducing novel technical challenges. Quantum complexity classes such as $\mathsf{BQP}$ (bounded-error quantum polynomial time), $\mathsf{QMA}$ (quantum Merlin-Arthur), and their variants have been extensively studied, revealing a rich landscape of quantum computational complexity.

In the quantum setting, proof systems take on new dimensions of complexity. Quantum proofs can exist in superposition states, verifiers can perform quantum computations, and the interaction between classical and quantum information becomes a central consideration. The study of quantum probabilistically checkable proofs, quantum interactive proofs, and quantum advice has revealed both similarities to and fundamental differences from their classical counterparts.

\subsection{Related works}
\paragraph{Quantum probabilistically checkable proofs and unique variants.}
    Our first result builds upon the Valiant-Vazirani theorem~\cite{ValiantVazirani1986}, which shows a randomized reduction from $\mathsf{NP}$ to its unique variant $\mathsf{UniqueNP}$. A quantum-classical analog was developed by Aharonov, Ben-Or, Brand{\~{a}}o and Sattath~\cite{Aharonov2022pursuitof}, who proved a reduction from $\mathsf{QCMA}$ to $\mathsf{UniqueQCMA}$ based on the same hashing technique. Our work generalizes this approach to establish a reduction from $\mathsf{QCPCP}$ to $\mathsf{UniqueQCPCP}$ under $\mathsf{BQ}$-operator and randomized reductions.

    The complexity class $\mathsf{QPCP}$ was first formally defined by Aharonov, Arad, Landau and Vazirani~\cite{qpcp_formulation_first}, though the proof-checking formulation has received limited attention since then. Recent work by Buhrman, Helsen, and Weggemans~\cite{BuhrmanHelsenWeggemans2024} and Buhrman, Le Gall, and Weggemans~\cite{BuhrmanLeGallWeggemans2024} has established foundational results for both $\mathsf{QPCP}$ and $\mathsf{QCPCP}$. The latter work introduces the $\mathsf{BQ}$-operator to formally capture the ``quantumness'' inherent in $\mathsf{QCPCP}$.

    In related developments, Irani, Natarajan, Nirkhe, Rao and Yuen~\cite{anandchinmaysearchtodecision} demonstrated that a single query to a classical $\mathsf{PP}$ oracle suffices to prepare a $\mathsf{QMA}$ witness to inverse-polynomial accuracy. They also proved that no polynomial-time reduction from $\mathsf{QMA}$-search to $\mathsf{QMA}$-decision exists relative to a suitable quantum oracle. Their ``state synthesis'' framework resolved Aaronson's question~\cite{Aaronson2016} regarding sublinear query complexity for arbitrary $n$-qubit states.

\paragraph{Quantum polynomial hierarchy and related complexity classes.}
    The quantum polynomial hierarchy was first defined by Yamakami~\cite{yamakami2003quantumnpquantumhierarchy}, where the verifier uses alternating quantum proofs. The standard description used in this paper differs from the original definition in that we consider quantum circuits rather than quantum Turing machines. Gharibian and Kempe~\cite{Gharibian_2012} subsequently studied $\mathsf{cq\text{-}\Sigma}_2$ (more commonly denoted $\mathsf{QC\Sigma}_2$) with universally quantified proofs. 

    Lockhart and Guillén~\cite{lockhart2017quantumstateisomorphism} introduced a hierarchy similar to $\mathsf{QCPH}$ but using existential and universal operators. While this hierarchy does not fully capture $\mathsf{QCMA}$, its structural properties are significantly easier to analyze than those of $\mathsf{QCPH}$, which we do not discuss in detail here.

    Significant progress on the $\mathsf{QCPH}$ hierarchy was made by Gharibian, Santha, Sikora, Sundaram and Yirka~\cite{qph_original} and Agarwal, Gharibian, Koppula and Rudolph~\cite{avantika}. \cite{qph_original} formally defined $\mathsf{QPH}$ and $\mathsf{QCPH}$ as used in this paper, and proved weaker variants of the Karp--Lipton theorem and Toda's theorem specifically for $\mathsf{QCPH}$. As a corollary, they showed that proving a Karp--Lipton theorem for $\mathsf{QPH}$ would imply $\mathsf{QMA}(2) \subseteq \mathsf{PSPACE}$, which appears beyond current techniques. We emphasize that we do not provide a standard Karp--Lipton statement for $\mathsf{QPH}$; rather, we prove a conditional inclusion collapse under the second level of $\mathsf{QPH}$ with polynomial-length quantum advice. They also established $\mathsf{QMA}(2) \subseteq \mathsf{Q\Sigma}_3 \subseteq \mathsf{NEXP}$, strengthening the upper bound on $\mathsf{QMA}(2)$.

    The collapse of $\mathsf{QCPH}$ was proven independently by Falor, Ge and Natarajan~\cite{falor2023collapsiblepolynomialhierarchypromise} and Agarwal, Gharibian, Koppula and Rudolph~\cite{avantika}. The latter work upper bounds $\mathsf{QCPH}$ by showing $\mathsf{QCPH} \subseteq \mathsf{pureQPH} \subseteq \mathsf{EXP^{PP}}$. Grewal and Yirka~\cite{grewaljustinqeph} demonstrated $\mathsf{QCPH} \subseteq \mathsf{QPH}$ at the expense of a constant factor blowup in the hierarchy level. 


    Agarwal, Gharibian, Koppula, and Rudolph~\cite{avantika} established a Karp--Lipton theorem for $\mathsf{QCPH}$ and provided an error reduction procedure for $\mathsf{pureQPH}$. In this work, we prove a non-uniform analogue of the Karp--Lipton theorem for $\mathsf{QPH}$. Separately, Grewal and Yirka~\cite{grewaljustinqeph} proved an unconditional collapse of $\mathsf{QEPH}$ to its second level, showing that $\mathsf{QEPH} = \mathsf{QRG}(1)$. Their analysis attributes the collapse to the ability of provers to exploit entanglement between rounds. We extend this perspective by identifying structural aspects of the protocol that also contribute to the collapse. In particular, we define bounded- and zero-entanglement variants of $\mathsf{QPH}$, denoted $\mathsf{BEQPH}$ and $\mathsf{Sep}\mathsf{QPH}$, and show that both classes retain certain computational limitations characteristic of $\mathsf{QEPH}$, as formalized in \cref{thm:beqph-collapse}. Grewal and Yirka~\cite{grewaljustinqeph} also showed that $\mathsf{QCPH} = \mathsf{DistributionQCPH}$, where the latter class permits distributions over proofs.

\subsection{Our contributions}
In this work, we investigate quantum analogues of two important variants of proof systems, contributing new insights to our understanding of quantum complexity theory. Our results span two main areas: quantum probabilistically checkable proofs and the quantum polynomial hierarchy.

\paragraph{Quantum-classical probabilistically checkable proofs.} 
    In~\cref{sec:uqcpcp_eq_qcpcp}, we examine the relationship between general quantum-classical PCP systems and their \textit{unique} variant, where the proof system is required to have a unique accepting quantum proof state. We demonstrate that restricting quantum-classical PCP systems to the unique variant does not reduce their computational power under $\mathsf{BQ}$-operator and randomized reductions. This establishes a collapse of general $\mathsf{QCPCP}$ systems to their uniqueness-restricted counterpart, providing new structural insights into quantum proof systems.

\paragraph{Quantum polynomial hierarchy.} 
    In~\cref{sec:qph}, we establish two results concerning the quantum polynomial hierarchy. First, we prove a non-uniform quantum analogue of the Karp–Lipton theorem: if $\mathsf{QMA} \subseteq \mathsf{BQP}/\mathsf{qpoly}$ (i.e., if every problem in quantum Merlin-Arthur can be solved by a bounded-error quantum polynomial-time machine with quantum polynomial advice), then the quantum polynomial hierarchy collapses under its second level with quantum advice: $\mathsf{QPH} \subseteq \mathsf{Q\Sigma_2}/\mathsf{qpoly}$. This result parallels the classical case while accounting for the unique features of quantum computation and quantum advice.
    
    Second, we define and study a variant of the quantum polynomial hierarchy motivated by the entangled-prover model considered in recent work on quantum entangled proofs~\cite{grewaljustinqeph}. This variant, which we refer to as the \emph{bounded-entanglement quantum polynomial hierarchy} ($\mathsf{BEQPH}$), imposes cross-round consistency conditions and additionally restricts the relative entropy of entanglement between the $i$-th register and the remaining registers. We show an unconditional collapse of the hierarchy at the fourth level: for all $k \geq 4$, $\mathsf{BEQ\Sigma}_k = \mathsf{BEQ\Sigma}_4$. Moreover, we show that all levels of $\mathsf{BEQPH}$ reduce to efficiently solvable convex optimization problems. We further extend these results to the class $\mathsf{Sep}\mathsf{QPH}$, defined as the zero-entanglement analogue of $\mathsf{BEQPH}$, and establish that it exhibits similar collapse and tractability properties across the hierarchy. These findings suggest that the tractability of $\mathsf{QEPH}$-like models arises not solely from entanglement, but also from the structural constraints imposed on the proof system.

\section{Equivalence of unique and standard quantum-classical PCPs}
\label{sec:uqcpcp_eq_qcpcp}
In this section, we establish our first collapse result by proving that $\mathsf{UniqueQCPCP} = \mathsf{QCPCP}$ under reductions in $\mathsf{BQ}+\mathsf{RP}$.

\subsection{Definitions and preliminaries}
\subsubsection{Complexity classes}
We formally define the unique variants we introduce, namely $\mathsf{UniqueQCPCP}$, $\mathsf{UniqueMTP}$, and $\mathsf{UniqueSAT}$.
\begin{definition}[Unique quantum-classical probabilistically checkable proof, $\UQCPCP$]
\label{def:unique-qcpcp}
Let $c, s$ be completeness and soundness parameters, and let $p, q : \mathbb{N} \to \mathbb{N}$, and $c, s : \mathbb{R}_{\geq 0} \to \mathbb{R}_{\geq 0}$. A promise problem $\mathcal{A} = (A_\mathrm{YES}, A_\mathrm{NO})$ is in $(p(n), q(n), c, s)$-$\UQCPCP$ if there exists a uniformly generated family of quantum circuits $\{V_n : n \in \mathbb{N}\}$ of polynomial size such that each $V_n$ receives an input $x \in \{0,1\}^n$ and a proof string $y \in \{0,1\}^{p(n)}$, from which at most $q(n)$ bits can be queried. The verifier accepts by measuring a designated output qubit and observing outcome $1$, and satisfies the following:

\begin{enumerate}[label=(\roman*)]
    \item \textnormal{(Completeness and Uniqueness)} If $x \in A_\mathrm{YES}$, then there exists a unique $y \in \{0,1\}^{p(n)}$ such that $\Pr[V_n(x, y) \textnormal{ accepts}] \geq c$, and for all $y' \neq y$, $\Pr[V_n(x, y') \textnormal{ accepts}] \leq s$.

    \item \textnormal{(Soundness)} If $x \in A_\mathrm{NO}$, then for all $y \in \{0,1\}^{p(n)}$, $\Pr[V_n(x, y) \textnormal{ accepts}] \leq s$.
\end{enumerate}
We write $\UQCPCP[q]$ if $p(n) = \textnormal{poly}(n)$, $c = 2/3$, and $s = 1/3$.
\end{definition}

\begin{definition}[Unique multilinear threshold problem, \texorpdfstring{$\UMTP$}{UMTP}]
\label{def:unique-mtp}
Let $P : \{0,1\}^N \to \mathbb{R}$ be a degree-$d$ multilinear polynomial with real coefficients $\{\beta_S\}_{S \subseteq [N], |S| \leq d}$. Suppose the image of $P$ lies in a finite set $D \subset [0,1]$ of evenly spaced values, such that $\log_2 |D| \leq \mathrm{poly}(N)$. Denote by $\delta_D$ the maximum element of $D$. Given a polynomial-size classical description of $P$ (i.e., the coefficients $\{\beta_S\}$), the set $D$, and a threshold $a \in D$, decide between the following two cases:
\begin{enumerate}[label=(\roman*)]
    \item There exists a unique $y \in \{0,1\}^N$ such that $P(y) \geq a$ and $P(y') < a$ for all $y' \neq y$.
    \item For all $y \in \{0,1\}^N$, it holds that $P(y) < a$.
\end{enumerate}
\end{definition}

\begin{definition}[Unique propositional satisfiability problem, $\mathsf{UniqueSAT}$]
\label{def:unique-sat}
Let $\phi$ be a Boolean formula over $n$ variables. Define the set of satisfying assignments as $\mathsf{SAT}(\phi) = \{x \in \{0,1\}^n : \phi(x) = 1\}$. Then the language $\mathsf{UniqueSAT} = \{\phi : |\mathsf{SAT}(\phi)| = 1\}$ consists of Boolean formulas that admit {exactly one} satisfying assignment.
\end{definition}

\subsubsection{Reduction techniques}
In this section, we present the formal definition of randomized reductions, as well as two formulations of the $\mathsf{BQ}$-operator as defined in~\cite{BuhrmanLeGallWeggemans2024}.
\begin{definition}[Randomized polynomial-time reduction]
\label{def:rp-reduction}
Let $\mathcal{A}$ and $\mathcal{B}$ be languages. A {randomized polynomial-time reduction} from $\mathcal{A}$ to $\mathcal{B}$ is a randomized algorithm $\mathcal{R}$ such that, for every input $x$:
\begin{enumerate}[label=(\roman*)]
  \item If $x \in \mathcal{A}$, then $\Pr[\mathcal{R}(x) \in B] \ge 1/2$.
  \item If $x \notin A$, then $\Pr[\mathcal{R}(x) \in B] = 0$.
\end{enumerate}
The algorithm $\mathcal{R}$ must run in time polynomial in $|x|$.
\end{definition}

\begin{lemma}[Valiant-Vazirani~{\cite{ValiantVazirani1986}}]
\label{lem:valiant-vazirani}
There exists a randomized polynomial-time algorithm $\mathcal{R} : \{\textnormal{CNF formulas over } n \textnormal{ variables}\} \rightarrow \{\textnormal{CNF formulas}\}$  such that, for every Boolean formula $\phi$ over $n$ variables, the following holds:
\begin{enumerate}[label=(\roman*)]
  \item \textnormal{(Completeness)} If $\phi$ is satisfiable, then $\Pr\left[|\mathsf{SAT}(\mathcal{R}(\phi))| = 1\right]\ge {1}/{(2n)}$.
    
  \item \textnormal{(Soundness)} If $\phi$ is unsatisfiable, then $\mathcal{R}(\phi)$ is unsatisfiable with probability $1$.
\end{enumerate}
Moreover, $\mathcal{R}$ runs in time polynomial in $|\phi|$.
\end{lemma}

\begin{definition}[\texorpdfstring{$\BQ$}{BQ}-operator~{\cite[Definition 6]{BuhrmanLeGallWeggemans2024}}]
\label{def:bq-1}
Let $\mathcal{C}$ be a class of promise problems. Define $\BQ \cdot \mathcal{C}$ as the class of all promise problems $\mathcal{A}$ for which there exists a promise problem $\mathcal{B} \in \mathcal{C}$ such that $\mathcal{A}$ reduces to $\mathcal{B}$ via a polynomial-time quantum reduction with success probability at least $2/3$: $\BQ \cdot \mathcal{C} := \{\mathcal{A} : \mathcal{A} \leq_q \mathcal{B} \textnormal{ for some } \mathcal{B} \in \mathcal{C}\}$.
\end{definition}

\begin{definition}[Alternative formulation of \texorpdfstring{$\BQ$}{BQ}-operator~{\cite[Definition 6a]{BuhrmanLeGallWeggemans2024}}]
\label{def:bq-2}
Let $\mathcal{C}$ be a class of promise problems. Then $\BQ \cdot \mathcal{C}$ consists of all promise problems $\mathcal{A} = (A_\mathrm{YES}, A_\mathrm{NO})$ such that there exists a promise problem $\mathcal{B} = (B_\mathrm{YES}, B_\mathrm{NO}) \in \mathcal{C}$ and a polynomial-time quantum algorithm $\mathcal{Q}$ with the following properties: for all $x$ with $|x| = n$,

\begin{enumerate}[label=(\roman*)]
    \item \textnormal{(Completeness)} If $x \in A_\mathrm{YES}$, then $\Pr_{z \sim \mathcal{Q}(x)}\left[(x, z) \in B_\mathrm{YES}\right] \geq 2/3$.

    \item \textnormal{(Soundness)} If $x \in A_\mathrm{NO}$, then $\Pr_{z \sim \mathcal{Q}(x)}\left[(x, z) \in B_\mathrm{NO}\right] \geq 2/3$.
\end{enumerate}
\end{definition}

We note that when the base class $\mathcal{C}$ contains at least the class $\mathsf{P}$, we adopt~\cref{def:bq-2} as the default interpretation of the $\BQ$-operator. Otherwise,~\cref{def:bq-1} is used. In this work,~\cref{def:bq-2} will suffice for all purposes.

\subsection{Reduction chain}
In~\cite{BuhrmanLeGallWeggemans2024}, a $\mathsf{BQ}$ reduction from $\mathsf{QCPCP}$ to $\mathsf{MTP}$ is presented in detail. In this section, we first show that $\mathsf{MTP}$ can be efficiently and deterministically reduced to $\mathsf{SAT}$. The reduction from $\mathsf{SAT}$ to $\mathsf{UniqueSAT}$ then follows directly from~\cite{ValiantVazirani1986} via an randomized reduction. Subsequently, we establish a reduction from $\mathsf{USAT}$ to $\mathsf{UniqueMTP}$, and finally a $\mathsf{BQ}$ reduction from $\mathsf{UniqueMTP}$ to $\mathsf{UniqueQCPCP}$. The overall proof structure is summarized in~\cref{fig:proof-chain}.

\begin{figure}[t!]
    \centering
    \includegraphics[width=0.9\linewidth]{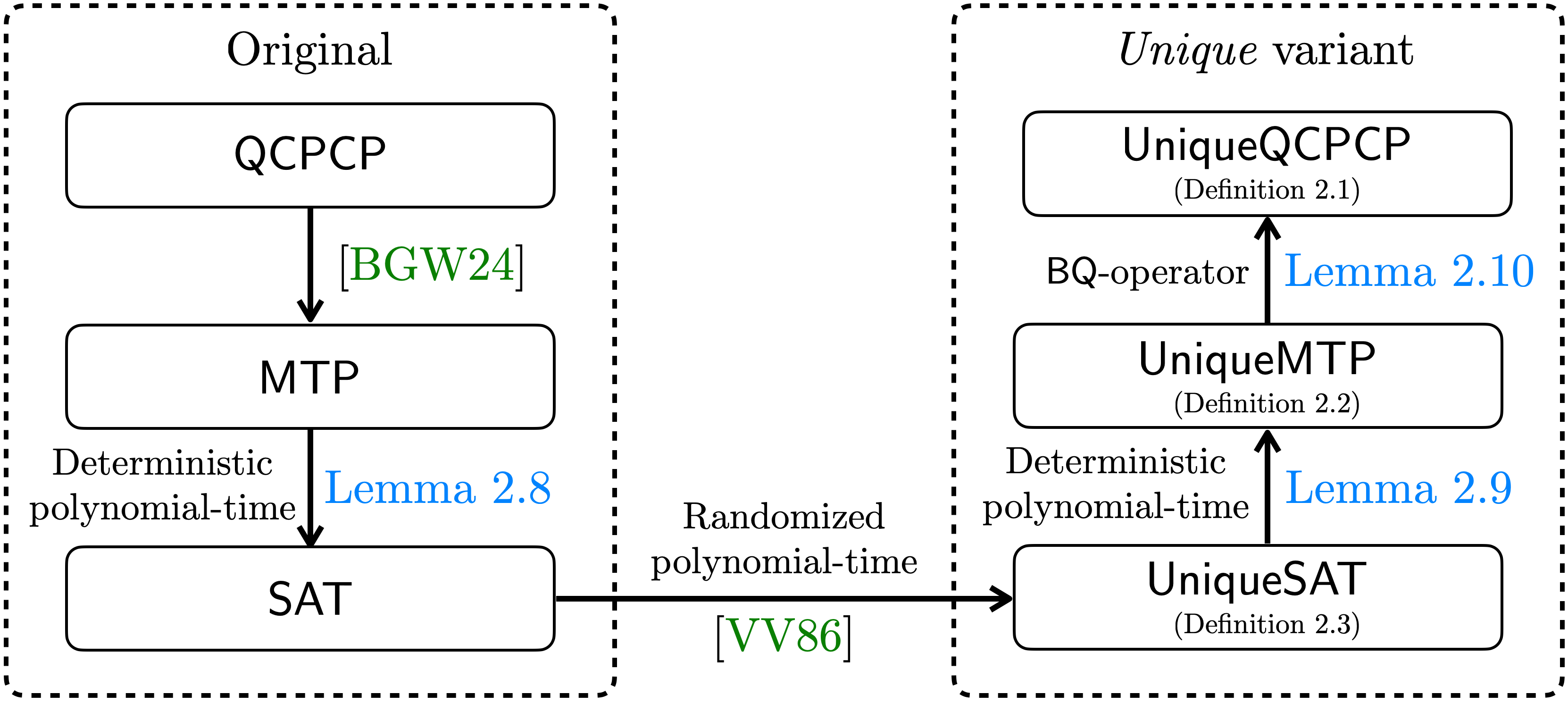}
    \caption{\textbf{Reduction structure used to prove the collapse result $\mathsf{UniqueQCPCP} = \mathsf{QCPCP}$.} The left side illustrates the original reduction chain from $\mathsf{QCPCP}$, as established in~\cite{BuhrmanLeGallWeggemans2024}. The right side depicts our refined reduction path to the unique variant, culminating in $\mathsf{UniqueQCPCP}$ (\cref{def:unique-qcpcp}). Key components include deterministic and randomized reductions, as well as the application of the $\mathsf{BQ}$-operator.}
    \label{fig:proof-chain}
\end{figure}

\begin{lemma}[Deterministic polynomial-time reduction from $\MTP$ to $\SAT$]\label{lem:mtp-sat}
    The $\MTP$ can be reduced to $\SAT$ deterministically and in polynomial time.
\end{lemma}

\begin{proof}
Let $P(y)$ denote the solution polynomial of the given $\MTP$ instance
\begin{equation}
    P(y) = \sum_{\substack{S \subseteq [p(n)],~|S| \leq 2q}} \beta_S \prod_{i \in S} y_i,
\end{equation}
where $\beta_S \in \mathbb{R}$ is represented using $\ell$-bit fixed-point encoding.

We assume the standard promise structure:
\begin{enumerate}[label=(\roman*)]
    \item (Correctness and Uniqueness) There exists at least a $y \in \{0,1\}^n$ such that $P(y) \geq a$.
    \item (Soundness) For all $y \in \{0,1\}^n$, $P(y) < a$.
\end{enumerate}

\paragraph{Transforming $P(y)$ into a Boolean formula.}
For each monomial $\prod_{i \in S} y_i$, define a clause:
\begin{equation}
    \phi_S := \bigwedge_{i \in S} y_i
\end{equation}
This is a conjunction over a subset of input variables. Each such clause $\phi_S$ can be encoded in CNF form via the Tseytin transformation.

\medskip\medskip

Each $\beta_S$ satisfies the bound (\cite[Lemma 3]{BuhrmanLeGallWeggemans2024}): $|\beta_S| \leq (1 + (1 + 2^{2q})^{2q - 1})\|P\|_{\infty}$, with $\|P\|_{\infty} \leq 1$. Thus, $\beta_S$ can be expressed as $v_S / 2^\ell$ with $v_S \in \mathbb{Z}$ represented using $k$-bit integers, where $k = \log_2 (1 + (1 + 2^{2q})^{2q - 1}) + \ell + 1$.

\medskip\medskip

Define the scaled polynomial:
\begin{equation}
    P'(y) := 2^\ell \cdot P(y) = \sum_{\substack{S \subseteq [p(n)],~|S| \leq 2q}} v_S \cdot \phi_S,
\end{equation}
where $v_S$ is a $k$-bit integer. Write $v_S$ in binary:
\begin{equation}
    v_S = \sum\limits_{i=1}^{k'} 2^i b_{S,i},
\end{equation}
with $b_{S,i} \in \{0,1\}$ and $k' = \log_2 (1 + (1 + 2^{2q})^{2q - 1})$.

\medskip\medskip

As $\phi_S \in \{0,1\}$, it acts like a Boolean variable. Thus:
\begin{align}
    P'(y) = \sum_{\substack{S \subseteq [p(n)],~ |S| \leq 2q}} \sum_i 2^i (b_{S,i} \land \phi_S) = \sum_S \left( \sum_i 2^i b_{S,i} \right) \cdot \phi_S.
\end{align}

\paragraph{Adder network construction.}
Each $v_S$ is represented in $k'$ bits. To compute $P'(y)$, we sum all such $k'$-bit numbers. This can be done using standard half-adder and full-adder networks, which can be encoded into CNF form using the Tseytin transformation. The total circuit has size polynomial in $|S| \cdot k'$ and logarithmic depth $\bigo(\log(|S| \cdot k'))$.

\paragraph{Threshold comparison circuit.}
We wish to enforce $\exists y\in\{0,1\}^n : P'(y) \geq 2^\ell a$, or in the soundness case, $\forall y\in\{0,1\}^n, \; P'(y) < 2^\ell a$. These threshold comparisons can be implemented by comparator circuits (in the class \textsf{CC}) and converted to CNF via Tseytin transformations.

\paragraph{Putting it all together.}
All parts of $P'(y)$—the monomials, the binary representations, the adder and comparator circuits—can be encoded into a Boolean formula in CNF with auxiliary variables $b_{S,i}$ and original variables $y_i$. This yields a deterministic and efficient reduction from $\MTP$ to $\SAT$.
\end{proof}

\begin{lemma}[Deterministic polynomial-time reduction from $\mathsf{UniqueSAT}$ to $\mathsf{UniqueMTP}$]
\label{lem:usat-to-umtp}
    There exists a deterministic polynomial-time reduction from $\mathsf{UniqueSAT}$ to $\mathsf{UniqueMTP}$.
\end{lemma}

\begin{proof}
    Given a CNF formula $\phi$ with clauses $C_1, C_2, \ldots, C_m$ over variables $x_1, x_2, \ldots, x_n$, we construct a multilinear polynomial as follows:
    \begin{equation}
        P(x_1, \ldots, x_n) = \sum_{j=1}^m \prod_{\ell \in C_j} \tilde{x}_\ell,
    \end{equation}
    where for each literal $\ell$ in clause $C_j$:
    \begin{equation}
        \tilde{x}_\ell = \begin{cases}
            x_i & \text{if } \ell = x_i, \\
            1 - x_i & \text{if } \ell = \neg x_i.
        \end{cases}
    \end{equation}

    \medskip\medskip
        
    We set the threshold to $a = m$.
    \begin{enumerate}[label=(\roman*)]
        \item (Completeness) If assignment $x$ satisfies $\phi$, then $x$ satisfies every clause $C_j$. For each clause $C_j$, at least one literal evaluates to true, making the corresponding product $\prod_{\ell \in C_j} \tilde{x}_\ell = 1$. Therefore, $P(x) = m \geq a$.
        \item (Soundness) If assignment $x$ does not satisfy $\phi$, then $x$ violates at least one clause $C_j$. For any violated clause, all literals evaluate to false, making the corresponding product $\prod_{\ell \in C_j} \tilde{x}_\ell = 0$. Thus, $P(x) \leq m - 1 < a$.
        \item (Uniqueness preservation) Since the reduction maps each satisfying assignment $x$ of $\phi$ to the same assignment $x$ for the polynomial $P$, uniqueness is preserved by the identity mapping.
    \end{enumerate}
    The reduction runs in polynomial time as constructing $P$ requires examining each clause and literal once.
\end{proof}

\begin{lemma}[\textsf{BQ} reduction from \(\mathsf{UniqueMTP}\) to \(\mathsf{UniqueQCPCP}\)]
\label{lem:umtptouqcpcp}
Let \(\Psi = (P(y),\, a,\, \delta_D)\) be an instance of the \(\mathsf{UniqueMTP}\) over Boolean variables \(y \in \{0,1\}^m\), where
\begin{equation}
    P(y) = \sum_{\substack{S \subseteq [m],\, |S| \le d}} \beta_S \prod_{i \in S} y_i,    
\end{equation}
and suppose that \(\Psi\) satisfies the following promise:
\begin{equation}
    \begin{cases}
        \exists!\, y^* \in \{0,1\}^m \textnormal{ such that } P(y^*) \ge a & \textnormal{(YES instance)}, \\
    \forall\, y \in \{0,1\}^m,\; P(y) < a & \textnormal{(NO instance)},
    \end{cases}    
\end{equation}
with a gap parameter \(\delta_D > 0\). Then there exists a polynomial-time deterministic quantum algorithm that outputs the description of a verifier $V_\Psi \in \mathsf{UniqueQCPCP}[p(m) = m,\, q = d,\, c,\, s]$, where
\[
c = \frac{a + \delta_D}{\sum_{j=1}^{N} |\beta_{S_j}|}
\quad \text{and} \quad
s = \frac{a - \delta_D}{\sum_{j=1}^{N} |\beta_{S_j}|},
\]
such that:

\begin{itemize}
  \item[(i)] If \(\Psi\) is a YES instance with unique witness \(y^*\), then \(V_\Psi\) accepts exactly on the proof \(y^*\) and rejects all other \(y\).
  \item[(ii)] If \(\Psi\) is a NO instance, then \(V_\Psi\) rejects every proof \(y\).
\end{itemize}
\end{lemma}

\begin{proof}
We describe the deterministic quantum reduction that maps \(\Psi\) to a verifier \(V_\Psi\). The following data is hardcoded into the description of \(V_\Psi\):
\[
\left\{S_j \subseteq [m] : |S_j| \le d \right\}, \quad
\{\beta_{S_j}\}_{j=1}^N, \quad 
a,\, \delta_D,
\]
where \(N = \sum_{i=0}^d \binom{m}{i} = \mathrm{poly}(m)\).

\paragraph{Verifier \(V_\Psi\).} On input proof \(y \in \{0,1\}^m\), the verifier proceeds as follows:

\begin{enumerate}[label=(\roman*)]
  \item Compute \(B := \sum_{j=1}^N |\beta_{S_j}|\) and sample a uniform real \(r \in [0, B)\) using \(\lceil \log (n) \rceil\) random bits.
  
  \item Find the unique index \(j\) such that
  \begin{equation}
    \sum_{i=1}^{j-1} |\beta_{S_i}| \le r < \sum_{i=1}^{j} |\beta_{S_i}|.    
  \end{equation}
  Let \(S = S_j\) and \(\sigma := \mathrm{sign}(\beta_{S_j}) \in \{\pm 1\}\).

  \item The verifier non-adaptively queries the qubits \(y_i\) for \(i \in S\) into an ancilla register via the unitary: $\ket{i}\ket{0} \mapsto \ket{i}\ket{y_i},~\forall i \in S$.    
  
  \item \textit{Phase kickback.} Apply a controlled-phase gate that multiplies the global phase by
  \begin{equation}
    \exp\left(i\pi \sigma \prod_{i \in S} y_i\right).
  \end{equation}

  \item \textit{Estimate contribution.} Use a Hadamard test subroutine on the ancilla to estimate the real part \(\Re[\beta_S \prod_{i \in S} y_i]\) up to additive error \(\delta_D/2\).

  \item To ensure this accuracy with high probability, repeat the Hadamard test subroutine \(T = O(1/\delta_D^2)\) times. This follows from Hoeffding's inequality: given bounded random variables \(X_1, \dots, X_T\) with \(|X_i| \le 1\) and mean \(\mu\), the empirical mean \(\hat{\mu} = \frac{1}{T} \sum_{i=1}^T X_i\) satisfies
  \begin{equation}
    \Pr[\, |\hat{\mu} - \mu| \ge \epsilon \,] \le 2 \exp(-2T\epsilon^2).    
  \end{equation}
  Setting \(\epsilon = \delta_D/2\), it suffices to take $T \ge {2\ln(2/\delta)}/{\delta_D^2}$ for failure probability \(\delta\). Since \(\delta_D\) is constant, this requires only a constant number of repetitions.

  \item \textit{Accept/reject.} Let \(\widehat{P}(y)\) be the empirical estimate of \(P(y)\) obtained by averaging the Hadamard test results over repeated monomial sampling. The verifier accepts if and only if \(\widehat{P}(y) \ge a\).

  To ensure correctness with respect to the \(\mathsf{QCPCP}\) gap, we normalize the thresholds by \(B = \sum_{j=1}^{N} |\beta_{S_j}|\), setting $c = {(a + \delta_D)}/{B}$, and $s = {(a - \delta_D)}/{B}$, which satisfies \(c - s = {2\delta_D}/{B} = {1}/{\mathrm{poly}(m)}\), sufficient for our \(\mathsf{QCPCP}\) setting.
\end{enumerate}

\paragraph{Analysis.} With this analysis, we conclude that $\mathsf{UniqueMTP} \le_{\mathsf{BQ}} \mathsf{UniqueQCPCP}$.

\begin{enumerate}[label=(\roman*)]
  \item (Bounded query) In each execution, \(V_\Psi\) queries at most \(|S| \le d\) bits of the proof \(y\), thus the query complexity is \(q = d\).
  
  \item (Completeness) If \(\Psi\) is a YES instance with unique satisfying assignment \(y^*\) such that \(P(y^*) \ge a\), then each monomial contributes exactly \(\beta_S \prod_{i \in S} y^*_i\), so \(\mathbb{E}[\widehat{P}(y^*)] = P(y^*) \ge a\). By concentration bounds, the empirical estimate \(\widehat{P}(y^*)\) exceeds \(a\) with high probability. No other \(y \ne y^*\) can satisfy \(P(y) \ge a\), hence only \(y^*\) is accepted.
  
  \item (Soundness) If \(\Psi\) is a NO instance, then for all \(y \in \{0,1\}^m\), we have \(P(y) < a\), so with high probability the estimate \(\widehat{P}(y)\) also satisfies \(\widehat{P}(y) < a\), and \(V_\Psi\) rejects all proofs.
  
  \item (Uniqueness preservation) The reduction \(\Psi \mapsto V_\Psi\) preserves the uniqueness of the satisfying assignment: the only proof that can be accepted by \(V_\Psi\) is the unique \(y\) such that \(P(y) \ge a\).
\end{enumerate}
This completes the proof.
\end{proof}

\begin{theorem}[Reduction from \(\mathsf{QCPCP}\) to \(\mathsf{UniqueQCPCP}\)]
\label{thm:qcpcp-to-uqcpcp}
Let $V$ be a $\mathsf{QCPCP}$ verifier for a promise problem $\mathcal{A} = (A_{\mathrm{YES}}, A_{\mathrm{NO}})$, with completeness $c$, soundness $s$, proof length $l_n$, and running time $T_n$. Then there exists a randomized polynomial-time algorithm which, given $V$, outputs a verifier $V'$ for a $\mathsf{UniqueQCPCP}$ system with the following properties:
\begin{enumerate}[label=(\roman*)]
    \item $V'$ runs in time $\mathcal{O}(T_n)$ and expects proofs of length $l_n + \mathcal{O}(\log (l_n))$.
    
    \item \textnormal{(Completeness):} If $x \in A_{\mathrm{YES}}$, then with probability at least $\Omega(1/\textnormal{poly}(l_n))$ over the randomness of the reduction, there exists a unique proof $y^* \in \{0,1\}^{l_n + \mathcal{O}(\log (l_n))}$ such that $\Pr[V'(x, y^*) = 1] \ge c - \gamma$, and $\forall y' \ne y^*,\; \Pr[V'(x, y') = 1] \le \gamma$, where $\gamma = \operatorname{negl}(n)$ is negligible in $n$.
    
    \item \textnormal{(Soundness)} If $x \in A_{\mathrm{NO}}$, then with probability $1$ over the randomness of the reduction, it holds that for all $y \in \{0,1\}^{l_n + \mathcal{O}(\log (l_n))}$, $\Pr[V'(x, y) = 1] \le s + \gamma$.
\end{enumerate}

\end{theorem}

\begin{proof}
This follows by composing the reductions:
\begin{equation}
    \mathsf{QCPCP} \le_{\mathsf{BQ}} \mathsf{MTP} \le_p \mathsf{SAT} \le_{\mathsf{RP}} \mathsf{UniqueSAT} \le_p \mathsf{UniqueMTP} \le_{\mathsf{BQ}} \mathsf{UniqueQCPCP},    
\end{equation}
where the first step is given by  of~\cite[Theorem 4]{BuhrmanLeGallWeggemans2024}, and the remaining steps are covered in~\cref{lem:mtp-sat},~Valiant-Vazirani reduction~{\cite{ValiantVazirani1986}},~\cref{lem:usat-to-umtp}, and~\cref{lem:umtptouqcpcp}. The only degradation in completeness and soundness arises from the randomized reduction from \(\mathsf{SAT}\) to \(\mathsf{UniqueSAT}\), which incurs a loss of at most \(1/l_n\) in each. This completes the proof.
\end{proof}

\subsection{Search-to-decision reductions}
For $\mathsf{QCMA}$, \cite{anandchinmaysearchtodecision} established a one-query Bernstein-Vazirani-style reduction using an oracle of the form $f(x) = x \cdot y^*$, where $y^*$ is the desired witness and $x \cdot y^*$ denotes the inner product. This reduction exploits the fact that $\mathsf{QCMA}$ verifiers have access to the entire classical proof string.

The situation for $\mathsf{QCPCP}$ is fundamentally different. A $\mathsf{QCPCP}$ oracle can only provide information of the form ``there exists a proof $y$ such that the verifier accepts,'' rather than direct access to inner products with proof bits. This limitation arises because $\mathsf{QCPCP}$ verifiers query only a constant number of bits from the proof, rather than having access to the entire proof string as in $\mathsf{QCMA}$.

Consequently, witness extraction for $\mathsf{QCPCP}$ appears to require bit-by-bit reconstruction, leading to $\mathcal{O}(n)$ oracle queries where $n$ is the proof length. With current techniques, no one-query Bernstein-Vazirani-style reduction analogous to that of~\cite{anandchinmaysearchtodecision} seems feasible for $\mathsf{QCPCP}$.

\section{Towards the collapse of the quantum polynomial hierarchy}
\label{sec:qph}
In this section, we investigate two perspectives on the potential collapse of the quantum polynomial hierarchy ($\QPH$). First, we present a non-uniform version of the Karp--Lipton theorem for $\QPH$. Second, we introduce a new class of problems that is structurally closer to the standard $\QPH$: namely, the proofs are restricted to product states (in contrast to the entangled states allowed in $\mathsf{QEPH}$), but are still interdependent due to a consistency condition enforced by the interaction of two players.

\subsection{Definitions and preliminaries}
\subsubsection{Complexity classes}
We begin by introducing the necessary definitions for quantum complexity classes including $\mathsf{Q\Sigma}_i(c,s)$, $\mathsf{QPH}$, $\mathsf{BQP}/\mathsf{qpoly}$ and $\mathsf{Q\Sigma}_2/\mathsf{qpoly}$.
\begin{definition}[$\mathsf{Q\Sigma}_i(c,s)$]
\label{def:qsigma-i}
Let $\mathcal{A} = (A_{\mathrm{YES}}, A_{\mathrm{NO}})$ be a promise problem, and let $c, s : \mathbb{N} \rightarrow [0,1]$ be polynomial-time computable functions. Then $\mathcal{A} \in \mathsf{Q\Sigma}_i(c, s)$ if there exists:
\begin{enumerate}[label=(\roman*)]
    \item a polynomially bounded function $p : \mathbb{N} \to \mathbb{N}$, and
    \item a polynomial-time uniform family of quantum circuits $\{V_n\}_{n \in \mathbb{N}}$,
\end{enumerate}
such that for every input $x \in \{0,1\}^n$, the verifier $V_n$ takes $p(n)$-qubit quantum states $\rho_1, \ldots, \rho_i$ as input proofs and outputs a single qubit. The verifier measures the output qubit in the computational basis, and the acceptance probability is defined on input $\rho_1 \otimes \cdots \otimes \rho_i$.

\begin{enumerate}[label=(\roman*)]
    \item \textnormal{(Completeness)} If $x \in L_{\mathrm{YES}}$, then $\exists \rho_1 \; \forall \rho_2 \; \cdots \; Q_i \rho_i$ such that $\Pr[V_n \textnormal{ accepts}] \geq c(n)$.    

    \item \textnormal{(Soundness)} If $x \in L_{\mathrm{NO}}$, then $\forall \rho_1 \; \exists \rho_2 \; \cdots \; \overline{Q}_i \rho_i$ such that $\Pr[V_n \textnormal{ accepts}] \leq s(n)$.
\end{enumerate}
Here, $Q_i = \forall$ if $i$ is even and $Q_i = \exists$ if $i$ is odd; $\overline{Q}_i$ denotes the complementary quantifier. When the completeness and soundness parameters are not specified, we define the class
\begin{equation}
    \mathsf{Q\Sigma}_i := \bigcup_{c-s \in \Omega(1/\mathrm{poly}(n))} \mathsf{Q\Sigma}_i(c,s).
\end{equation}
\end{definition}

\begin{definition}[Quantum polynomial hierarchy, $\mathsf{QPH}$]
\label{def:qph}
The quantum polynomial hierarchy $\mathsf{QPH}$ is defined as
\begin{equation}
    \mathsf{QPH} := \bigcup_{i \in \mathbb{N}} \mathsf{Q\Sigma}_i = \bigcup_{i \in \mathbb{N}} \mathsf{Q\Pi}_i,
\end{equation}
where $\mathsf{Q\Pi}_i := \mathsf{co\text{-}Q\Sigma}_i$.
\end{definition}

\begin{definition}[$\mathsf{BQP}/\mathsf{qpoly}$]
\label{def:bqp-qpoly}
A promise problem $\mathcal{A} = (A_{\mathrm{YES}}, A_{\mathrm{NO}})$ is in $\mathsf{BQP}/\mathsf{qpoly}$ if there exist:
\begin{enumerate}[label=(\roman*)]
    \item a polynomial-size family of quantum circuits $\{C_n\}_{n \in \mathbb{N}}$, and
    \item a sequence of polynomial-size quantum advice states $\{|\psi_n\rangle\}_{n \in \mathbb{N}}$,
\end{enumerate}
such that for all $x \in \{0,1\}^n$, the circuit $C_n$ takes input $|x\rangle \otimes |\psi_n\rangle$ and satisfies:
\begin{enumerate}[label=(\roman*)]
    \item \textnormal{(Completeness)} If $x \in A_{\mathrm{YES}}$, then $\Pr[C_n(|x\rangle \otimes |\psi_n\rangle) = 1] \geq 2/3$.
    
    \item \textnormal{(Soundness)} If $x \in A_{\mathrm{NO}}$, then $\Pr[C_n(|x\rangle \otimes |\psi_n\rangle) = 1] \leq 1/3$.
\end{enumerate}
Note that $\mathsf{BQP}/\mathsf{qpoly}$ is the quantum analogue of $\mathsf{BPP}/\mathsf{poly}$, where quantum advice states replace classical advice strings.
\end{definition}

\begin{definition}[$\mathsf{Q\Sigma}_2/\mathsf{qpoly}$]
\label{def:qsigma2-qpoly}
A promise problem $\mathcal{A} = (A_{\mathrm{YES}}, A_{\mathrm{NO}})$ is in $\mathsf{Q\Sigma}_2/\mathsf{qpoly}(c,s)$ if there exist:
\begin{enumerate}[label=(\roman*)]
    \item a polynomial-time uniform family of quantum circuits $\{V_n\}_{n \in \mathbb{N}}$,
    \item a sequence of polynomial-size quantum advice states $\{|\psi_n\rangle\}_{n \in \mathbb{N}}$, and
    \item polynomial-time computable functions $c, s : \mathbb{N} \rightarrow [0,1]$,
\end{enumerate}
such that for all $x \in \{0,1\}^n$, the verifier $V_n$ takes $(x, |\psi_n\rangle, \rho_1, \rho_2)$ as input, where $\rho_1, \rho_2$ are polynomial-size quantum states, and satisfies:
\begin{enumerate}[label=(\roman*)]
    \item \textnormal{(Completeness)} If $x \in A_{\mathrm{YES}}$, then $\exists \rho_1 \; \forall \rho_2$ such that $\Pr[V_n(x, |\psi_n\rangle, \rho_1, \rho_2) = 1] \geq c(n)$.    
    
    \item \textnormal{(Soundness)} If $x \in A_{\mathrm{NO}}$, then $\forall \rho_1 \; \exists \rho_2$ such that $\Pr[V_n(x, |\psi_n\rangle, \rho_1, \rho_2) = 1] \leq s(n)$.
\end{enumerate}
When the completeness-soundness gap is not explicitly given, we define
\begin{equation}
    \mathsf{Q\Sigma}_2/\mathsf{qpoly} := \bigcup_{c-s \in \Omega(1/\mathrm{poly}(n))} \mathsf{Q\Sigma}_2/\mathsf{qpoly}(c,s).    
\end{equation}
\end{definition}

\begin{definition}[Quantum entangled polynomial hierarchy, $\mathsf{QEPH}$~\cite{grewaljustinqeph}]
\label{def:qeph}
The class $\mathsf{QEPH}$ is defined analogously to $\mathsf{QPH}$, except that the proof states $\rho_1, \ldots, \rho_i$ in the completeness and soundness conditions are allowed to be arbitrary entangled states over all proof registers, rather than being restricted to product states as in $\mathsf{QPH}$.
\end{definition}

\begin{remark}[Containment relations]
It is straightforward to verify that $\mathsf{QPH} \subseteq \mathsf{QEPH}$ and $\mathsf{BQP}/\mathsf{qpoly} \subseteq \mathsf{Q\Sigma}_2/\mathsf{qpoly}$. The relationship between quantum and classical advice in the hierarchy setting (e.g., whether $\mathsf{Q\Sigma}_i/\mathsf{qpoly} \subseteq \mathsf{Q\Sigma}_i/\mathsf{poly}$) remains an interesting open question.
\end{remark}

\subsubsection{Convexity and entanglement measures}

In this section, we collect key convexity-theoretic results that underpin our collapse arguments. In particular, we highlight structural properties of the state space of quantum systems and their implications for optimization over separable states.

\begin{theorem}[Sion's minimax theorem~\cite{sion1958minimax}]
\label{thm:sion}
Let $X$ and $Y$ be topological vector spaces, and let $A \subseteq X$ and $B \subseteq Y$ be convex and compact subsets. Suppose $f : A \times B \rightarrow \mathbb{R}$ is a function that is concave in its first argument and convex in its second argument. Then,
\begin{equation}
    \max_{a \in A} \min_{b \in B} f(a, b) = \min_{b \in B} \max_{a \in A} f(a, b).
\end{equation}
In the setting of quantum proof systems, the set of density matrices $\mathcal{D}(\mathcal{H})$ over a finite-dimensional Hilbert space $\mathcal{H}$ forms a convex and compact subset under the trace norm topology~\cite{watrous2018quantum}.
\end{theorem}

\begin{theorem}[Krein--Milman theorem~\cite{Krein1940}]
\label{thm:krein-milman}
Let \( K \) be a non-empty compact convex subset of a locally convex topological vector space \( X \). Then, $K = \overline{\mathrm{conv}}(\operatorname{ext}(K))$, where \( \operatorname{ext}(K) \) denotes the set of extreme points of \( K \), and \( \overline{\mathrm{conv}} \) denotes the closure of the convex hull of these extreme points with respect to the topology of \( X \).
\end{theorem}

\begin{definition}[Relative entropy of entanglement]\label{def:ree}
The relative entropy of entanglement of a bipartite quantum state \( \rho_{AB} \) is defined as
\begin{equation}
    \mathcal{E}(\rho_{AB}) = \inf_{\sigma_{AB} \in \mathbf{Sep}} S(\rho_{AB} \| \sigma_{AB}),
\end{equation}
where \( \mathbf{Sep} \) denotes the set of separable states and \( S(\rho \| \sigma) = \operatorname{Tr}[\rho \log (\rho)] - \operatorname{Tr}[\rho \log (\sigma)] \) is the quantum relative entropy between \( \rho \) and \( \sigma \).
\end{definition}

\subsection{Non-uniform Karp--Lipton collapse for quantum polynomial hierarchy}
In this section, we prove the following quantum Karp--Lipton theorem by leveraging the quantum complexity classes defined in the previous section along with key convexity properties.
\begin{theorem}[Quantum Karp--Lipton]
\label{thm:quantum-Karp--Lipton}
If $\mathsf{QMA} \subseteq \mathsf{BQP}/\mathsf{qpoly}$, then the quantum polynomial hierarchy collapses to its second level with quantum advice:
\begin{equation}
    \mathsf{QPH} = \bigcup_{k \in \mathbb{N}} \mathsf{Q\Sigma}_k \subseteq \mathsf{Q\Sigma}_2/\mathsf{qpoly}.
\end{equation}
\end{theorem}

\begin{proof}
We prove the result by induction on $k$. Throughout the proof, we use the standard completeness and soundness parameters $c = 2/3$ and $s = 1/3$.

\paragraph{Base Case ($k = 3$).}  
Let $\mathcal{A} \in \mathsf{Q\Sigma_3}$. Then by~\cref{def:qsigma-i}, there exists a polynomial-time uniform family of quantum circuits $\{V_n^{(3)}\}$ such that for all $x \in \{0,1\}^n$:
\begin{align}
x \in A_{\mathrm{YES}} &\iff \exists \rho_1 \, \forall \rho_2 \, \exists \rho_3 : \Pr[V_n^{(3)}(x; \rho_1, \rho_2, \rho_3) = 1] \geq 2/3, \\
x \in A_{\mathrm{NO}} &\iff \forall \rho_1 \, \exists \rho_2 \, \forall \rho_3 : \Pr[V_n^{(3)}(x; \rho_1, \rho_2, \rho_3) = 1] \leq 1/3.
\end{align}

\medskip\medskip

Define a quantum promise problem $\mathsf{P}_x$ with input $(\rho_1, \rho_2)$ as follows:
\begin{align}
(\rho_1, \rho_2) \in \mathsf{P}_{x,\text{YES}} &\iff \exists \rho_3 : \Pr[V_n^{(3)}(x; \rho_1, \rho_2, \rho_3) = 1] \geq 2/3, \\
(\rho_1, \rho_2) \in \mathsf{P}_{x,\text{NO}} &\iff \forall \rho_3 : \Pr[V_n^{(3)}(x; \rho_1, \rho_2, \rho_3) = 1] \leq 1/3.
\end{align}

\medskip\medskip

By construction, $\mathsf{P}_x \in \mathsf{QMA}$ for each fixed $x$. By the assumption $\mathsf{QMA} \subseteq \mathsf{BQP}/\mathsf{qpoly}$, there exists a polynomial-time quantum circuit family $\{C_n\}$ and a polynomial-size quantum advice sequence $\{|\alpha_n\rangle\}$ (depending only on $n = |x|$) such that:
\begin{align}
(\rho_1, \rho_2) \in \mathsf{P}_{x,\text{YES}} &\implies \Pr[C_n(x, \rho_1, \rho_2, |\alpha_n\rangle) = 1] \geq 2/3, \\
(\rho_1, \rho_2) \in \mathsf{P}_{x,\text{NO}} &\implies \Pr[C_n(x, \rho_1, \rho_2, |\alpha_n\rangle) = 1] \leq 1/3.
\end{align}
Substituting this characterization into the original quantifier expressions:
\begin{align}
x \in A_{\mathrm{YES}} &\iff \exists \rho_1 \, \forall \rho_2 : \Pr[C_n(x, \rho_1, \rho_2, |\alpha_n\rangle) = 1] \geq 2/3, \\
x \in A_{\mathrm{NO}} &\iff \forall \rho_1 \, \exists \rho_2 : \Pr[C_n(x, \rho_1, \rho_2, |\alpha_n\rangle) = 1] \leq 1/3.
\end{align}
This defines a $\mathsf{Q\Sigma_2}$ protocol with quantum advice $|\alpha_n\rangle$. Therefore, $\mathcal{A} \in \mathsf{Q\Sigma_2}/\mathsf{qpoly}$.

\paragraph{Inductive Step.}  
Assume the conditional inclusion holds for all $\mathsf{Q\Sigma}_j$ with $j \leq k-1$, i.e., $\mathsf{Q\Sigma}_j \subseteq \mathsf{Q\Sigma_2}/\mathsf{qpoly}$. We prove that $\mathsf{Q\Sigma}_k \subseteq \mathsf{Q\Sigma_2}/\mathsf{qpoly}$ for $k \geq 4$.

\medskip\medskip

Let $\mathcal{A} \in \mathsf{Q\Sigma}_k$. By definition:
\begin{equation}
    x \in A_{\mathrm{YES}} \iff \exists \rho_1 \, \forall \rho_2 \, \exists \rho_3 \cdots Q_k \rho_k : \Pr[V_n^{(k)}(x; \rho_1, \ldots, \rho_k) = 1] \geq 2/3.
\end{equation}
We can restructure this as:
\begin{equation}
    x \in A_{\mathrm{YES}} \iff \exists \rho_1 \, \forall \rho_2 : (x, \rho_1, \rho_2) \in \mathcal{A}',    
\end{equation}
where the quantum language $\mathcal{A}' \in \mathsf{Q\Sigma}_{k-2}$ is defined by:
\begin{equation}
    (x, \rho_1, \rho_2) \in A'_{\mathrm{YES}} \iff \exists \rho_3 \cdots Q_k \rho_k : \Pr[V_n^{(k)}(x; \rho_1, \ldots, \rho_k) = 1] \geq 2/3.
\end{equation}

\medskip\medskip

By the induction hypothesis, $\mathcal{A}' \in \mathsf{Q\Sigma_2}/\mathsf{qpoly}$. Therefore, there exists a verifier $N_n$ and advice state $|\beta_n\rangle$ such that:
\begin{equation}
    (x, \rho_1, \rho_2) \in A'_{\mathrm{YES}} \iff \exists \sigma_1 \, \forall \sigma_2 : \Pr[N_n(x, \rho_1, \rho_2, \sigma_1, \sigma_2, |\beta_n\rangle) = 1] \geq 2/3.    
\end{equation}
Substituting back:
\begin{equation}
    x \in A_{\mathrm{YES}} \iff \exists \rho_1 \, \forall \rho_2 \, \exists \sigma_1 \, \forall \sigma_2 :
    \Pr[N_n(x, \rho_1, \rho_2, \sigma_1, \sigma_2, |\beta_n\rangle) = 1] \geq 2/3.    
\end{equation}

\medskip\medskip

The inner three quantifiers $\forall\rho_2\exists\sigma_1\forall\sigma_2$ form a $\mathsf{Q\Pi_3/qpoly}$ predicate. Using base case and $\mathsf{Q\Sigma_2=Q\Pi_2}$, we have 
\begin{align}
    \mathsf{Q\Pi}_3/\mathsf{qpoly}&=\mathsf{co\text{-}(Q\Sigma}_3/\mathsf{qpoly)}\subseteq\mathsf{co\text{-}(Q\Sigma}_2/\mathsf{qpoly)}\\
    &=\mathsf{Q\Pi}_2/\mathsf{qpoly}=\mathsf{Q\Sigma}_2/\mathsf{qpoly}.
\end{align}
Thus this predicate can be squeezed to a $\mathsf{Q\Sigma_2/qpoly}$ as follows, where $M_n$ is a $\mathsf{Q\Sigma_2/qpoly}$ verifier and quantum promise language $\mathcal{A}''\in\mathsf{Q\Pi_3/qpoly}$.
\begin{equation}
    (x,\rho_1)\in A_\textnormal{YES}'' \iff \forall\rho_2\exists\sigma_1\forall\sigma_2 : \operatorname{Pr}[N_n(x,\rho_1,\rho_2,\sigma_1,\sigma_2,\ket{\beta_n})=1]\geq2/3.
\end{equation}
\begin{equation}
    (x,\rho_1)\in A''_\textnormal{YES} \iff \exists\rho_1'\forall\sigma_1' : \operatorname{Pr}[M_n(x,\rho_1,\rho_1',\sigma_1',\ket{\gamma_n})=1] \geq 2/3.    
\end{equation}
Substituting back,
\begin{equation}
    x \in A_\textnormal{YES} \iff \exists\rho_1\exists\rho_1'\forall\sigma_1' : \operatorname{Pr}[M_n(x,\rho_1,\rho_1',\sigma_1',\ket{\gamma_n})=1]\geq2/3,    
\end{equation}
which is a $\mathsf{Q\Sigma_2/qpoly}$ predicate with existential proof $\rho_1\otimes\rho_1'$ and universal proof $\sigma_1'$. The soundness condition follows by similar reasoning.

\medskip\medskip
Hence, $\mathsf{Q\Sigma}_k \subseteq \mathsf{Q\Sigma_2}/\mathsf{qpoly}$ for all $k \geq 3$, and the result follows for $k \leq 2$ trivially.
\end{proof}


\subsection{Bounded-entanglement quantum polynomial hierarchy}

We introduce a variant of the quantum polynomial hierarchy that enforces structural consistency across rounds of interaction, along with an upper bound on the entanglement of the proof state as measured by the von Neumann entropy.

\begin{definition}[Consistent extension property]
\label{def:consistent-extension}
Let $\mathcal{H}_1, \mathcal{H}_2, \ldots, \mathcal{H}_k$ be finite-dimensional Hilbert spaces. A sequence of quantum states $\{\rho_i\}_{i=1}^k$ satisfies the \emph{consistent extension property} if, for each $i \geq 2$, there exists a state $\rho_i$ on $\mathcal{H}_1 \otimes \cdots \otimes \mathcal{H}_i$ such that $\mathrm{Tr}_{\mathcal{H}_i}(\rho_i) = \rho_{i-1}$.
\end{definition}

In a $\mathsf{BEQPH}$ protocol, each prover is required to submit states that satisfy the \emph{consistent extension property}. Specifically, if the prover's $i$-th message is a quantum state on registers $\mathsf{A}_1 \otimes \cdots \otimes \mathsf{A}_i$, then tracing out register $\mathsf{A}_i$ must yield exactly the $(i{-}1)$-th message. In addition, the von Neumann entropy at each round $i$ is upper bounded by a constant $B_i$, where $b_{i-1} < B_{i-1}$. Here, $B_{i-1}$ denotes the maximum entanglement entropy achievable between register $\mathsf{A}_i$ and the joint system $\mathsf{A}' := \mathsf{A}_1 \otimes \cdots \otimes \mathsf{A}_{i-1}$, upon the introduction of $\mathsf{A}_i$ at round $i$. We measure the entanglement between $\mathsf{A}_i$ and $\mathsf{A}'$ using the relative entropy of entanglement, as defined in~\cref{def:ree}.

\begin{definition}[Bounded-entanglement quantum polynomial hierarchy, $\mathsf{BEQPH}^{\mathbf{b}, \mathbf{d}}$]
\label{def:beqph}
Let $c, s : \mathbb{N} \to [0,1]$ be polynomial-time computable functions satisfying $c(n) - s(n) \geq 1/\mathrm{poly}(n)$, where $n = |x|$. A promise problem $\mathcal{A} = (A_{\mathrm{YES}}, A_{\mathrm{NO}})$ belongs to $\mathsf{BEQ\Sigma}_k^{\mathbf{b},\mathbf{d}}(c,s)$ if there exists a polynomial-time uniform family of quantum circuits $\{V_n\}$ such that:
\begin{enumerate}[label=(\roman*)]
    \item \textnormal{(Completeness)} For all $x \in A_{\mathrm{YES}}$, there exists a sequence of quantum proofs satisfying the consistent extension property such that $\exists \rho_1 \, \forall \sigma_1 \, \exists \rho_2 \, \forall \sigma_2 \, \cdots \, Q_k \tau_k :
        \Pr[V_n(\tau_{k-1}, \tau_k) = 1] \geq c(n)$.
    
    \item \textnormal{(Soundness)} For all $x \in A_{\mathrm{NO}}$, any sequence of proofs satisfying the consistent extension property must satisfy $\forall \rho_1 \, \exists \sigma_1 \, \forall \rho_2 \, \exists \sigma_2 \, \cdots \, \overline{Q}_k \tau_k :
        \Pr[V_n(\tau_{k-1}, \tau_k) = 1] \leq s(n)$.
\end{enumerate}
Here, $\rho_j$ and $\sigma_j$ are chosen from the admissible sets
\[
S_j :=
\begin{cases}
    \left\{
    \rho \in \mathcal{D}(\mathsf{A}_1 \otimes \cdots \otimes \mathsf{A}_j)
    \;\middle|\;
    \begin{array}{l}
    \mathrm{Tr}_{\mathsf{A}_j}(\rho) = \rho_{j-1}, \\[2pt]
    \mathcal{E}_{\mathsf{A}_j, \mathsf{A}'_j}(\rho) \leq b_{j-1} < \max \mathcal{E}_{\mathsf{A}_j, \mathsf{A}'_j}(\rho)
    \end{array}
    \right\} & \textnormal{(Prover A)} \\[12pt]
    
    \left\{
    \sigma \in \mathcal{D}(\mathsf{B}_1 \otimes \cdots \otimes \mathsf{B}_j)
    \;\middle|\;
    \begin{array}{l}
    \mathrm{Tr}_{\mathsf{B}_j}(\sigma) = \sigma_{j-1}, \\[2pt]
    \mathcal{E}_{\mathsf{B}_j, \mathsf{B}'_j}(\sigma) \leq d_{j-1} < \max \mathcal{E}_{\mathsf{B}_j, \mathsf{B}'_j}(\sigma)
    \end{array}
    \right\} & \textnormal{(Prover B)}
\end{cases}
\]
with $\mathsf{A}'_j := \mathsf{A}_1 \otimes \cdots \otimes \mathsf{A}_{j-1}$ and $\mathsf{B}'_j := \mathsf{B}_1 \otimes \cdots \otimes \mathsf{B}_{j-1}$. The quantifier $Q_k$ is $\exists$ if $k$ is odd and $\forall$ if $k$ is even. All proof states must satisfy the consistent extension property at each round and adhere to the entanglement bounds $\mathbf{b} = \{b_j\}_{j=1}^{k}$ and $\mathbf{d} = \{d_j\}_{j=1}^{k}$. We define
\begin{equation*}
    \mathsf{BEQ\Sigma}^{\mathbf{b},\mathbf{d}}_k := \bigcup_{c,s} \mathsf{BEQ\Sigma}^{\mathbf{b},\mathbf{d}}_k(c,s)
    \quad \text{and} \quad
    \mathsf{BEQPH}^{\mathbf{b},\mathbf{d}} := \bigcup_{k=1}^{\infty} \mathsf{BEQ\Sigma}^{\mathbf{b},\mathbf{d}}_k.
\end{equation*}
In what follows, we will omit the superscripts $\mathbf{b}, \mathbf{d}$ whenever they are clear from context.
\end{definition}

\begin{definition}[Separable quantum polynomial hierarchy]
\label{def:sepqph}
The separable quantum polynomial hierarchy, denoted by $\mathsf{SepQPH}$, is the special case of $\mathsf{BEQPH}^{\mathbf{b}, \mathbf{d}}$ in which all proof states are required to be separable. That is, $\mathsf{SepQPH} := \mathsf{BEQPH}^{\mathbf{0}, \mathbf{0}}$, where $b_i = d_i = 0$ for all $i$.
\end{definition}


\begin{theorem}[Collapse of the bounded-entanglement quantum polynomial hierarchy]
\label{thm:beqph-collapse}
Each level of $\mathsf{BEQPH}$ collapses to a two-round entanglement-bounded convex optimization problem. Moreover, for all $i > 4$, it holds that $\mathsf{BEQ\Sigma}_i = \mathsf{BEQ\Sigma}_4$.
\end{theorem}

\begin{proof}
To show that each level of $\mathsf{BEQPH}$ reduces to a two-round convex optimization problem, we begin by analyzing the case of $\mathsf{BEQ\Sigma}_3$, and then generalize to higher levels.

\medskip

Let the acceptance probability of a $\mathsf{BEQ\Sigma}_3$ protocol be given by the min-max expression
\begin{equation}
\label{eq:mixed_optim_prob}
v^*_3 = \max_{\rho_1 \in \mathcal{D}(\mathsf{A}_1)} \min_{\sigma_1 \in \mathcal{D}(\mathsf{B}_1)} \max_{\rho_2 \in \mathcal{S}(\rho_1)} \operatorname{Tr}(R(\rho_2 \otimes \sigma_1)),
\end{equation}
where $\mathcal{S}(\rho_1)$ denotes the set of consistent extensions of $\rho_1$ under bounded entanglement:
\begin{equation}
    \mathcal{S}(\rho_1) := \left\{ \rho \in \mathcal{D}(\mathsf{A}_1 \otimes \mathsf{A}_2) \;\middle|\; \operatorname{Tr}_{\mathsf{A}_2}(\rho) = \rho_1,\; \mathcal{E}_{\mathsf{A}_2, \mathsf{A}_1}(\rho) \leq b_2 < \max \mathcal{E}_{\mathsf{A}_2, \mathsf{A}_1}(\rho) \right\}.    
\end{equation}

In the absence of entanglement bounds, this reduces to a $\mathsf{QEPH}$ protocol, which is known to collapse to its second level~\cite{grewaljustinqeph}. The set $\mathcal{S}(\rho_1)$ is convex and compact: convexity follows from the convexity of relative entropy, and compactness follows from standard arguments in~\cite{grewaljustinqeph}.

\medskip

Hence, by Sion's minimax theorem (\cref{thm:sion}), the quantifiers can be swapped:
\begin{equation}
    v^*_3 = \max_{\rho_1} \min_{\sigma_1} \max_{\rho_2 \in \mathcal{S}(\rho_1)} \operatorname{Tr}(R(\rho_2 \otimes \sigma_1))
= \max_{\rho_1} \max_{\rho_2 \in \mathcal{S}(\rho_1)} \min_{\sigma_1} \operatorname{Tr}(R(\rho_2 \otimes \sigma_1)).    
\end{equation}
Defining $\mathcal{T}^{(1)}_{\rho,\mathbf{b}} := \left\{ \rho \in \mathcal{D}(\mathsf{A}_1 \otimes \mathsf{A}_2) \;\middle|\; \mathcal{E}_{\mathsf{A}_1, \mathsf{A}_2}(\rho) \leq b_1 < \max \mathcal{E}_{\mathsf{A}_1, \mathsf{A}_2}(\rho) \right\}$, we obtain the simplified form:
\begin{equation}
    v^*_3 = \max_{\rho_2 \in \mathcal{T}^{(1)}_{\rho,\mathbf{b}}} \min_{\sigma_1 \in \mathcal{D}(\mathsf{B}_1)} \operatorname{Tr}(R(\rho_2 \otimes \sigma_1)).    
\end{equation}
This is a single-round convex optimization problem with convex domains and a linear objective, hence efficiently solvable via semidefinite programming.

\medskip

A similar argument yields the expression for $\mathsf{BEQ\Sigma}_4$:
\begin{equation}\label{eq:4level}
    v^* = \max_{\rho_1} \min_{\sigma_1} \max_{\rho_2 \in \mathcal{S}(\rho_1)} \min_{\sigma_2 \in \mathcal{S}(\sigma_1)} \operatorname{Tr}(R(\rho_2 \otimes \sigma_2))
= \max_{\rho_2 \in \mathcal{T}^{(2)}_{\rho,\mathbf{b}}} \min_{\sigma_2 \in \mathcal{T}^{(2)}_{\sigma,\mathbf{d}}} \operatorname{Tr}(R(\rho_2 \otimes \sigma_2)),    
\end{equation}
where the sets are defined as
\begin{align}
    \mathcal{T}^{(1)}_{\rho,\mathbf{b}} &:= \left\{ \rho \in \mathcal{D}(\mathsf{A}_1 \otimes \mathsf{A}_2) \;\middle|\; \mathcal{E}_{\mathsf{A}_1, \mathsf{A}_2}(\rho) \leq b_1 \right\},    \\
    \mathcal{T}^{(1)}_{\sigma,\mathbf{d}} &:= \left\{ \sigma \in \mathcal{D}(\mathsf{B}_1 \otimes \mathsf{B}_2) \;\middle|\; \mathcal{E}_{\mathsf{B}_1, \mathsf{B}_2}(\sigma) \leq d_1 \right\}.
\end{align}

\medskip

For general $i > 2$, the $(i+2)$-th level of $\mathsf{BEQPH}$ reduces to the following convex-concave bilinear minimax problem:
\begin{equation}
\label{eq:general_i}
v^*_{i+2} = \max_{\rho \in \mathcal{T}^{(k)}_{\rho,\mathbf{b}}} \min_{\sigma \in \mathcal{T}^{(l)}_{\sigma,\mathbf{d}}} \operatorname{Tr}\left( R(\rho \otimes \sigma) \right),
\end{equation}
where the parameters $k$ and $l$ are given by
\[
k = l = \frac{i}{2} \quad \text{if } (i+2) \text{ is even}, \qquad
k = \frac{i+1}{2}, \; l = k - 1 \quad \text{if } (i+2) \text{ is odd}.
\]
The feasible set $\mathcal{T}^{(m)}_{\tau,\mathbf{c}}$ denotes the collection of density operators over $\mathsf{R}_1 \otimes \cdots \otimes \mathsf{R}_{m+1}$ with bounded relative entropy of entanglement between the last register $\mathsf{R}_{m+1}$ and the rest:
\[
\mathcal{T}^{(m)}_{\tau,\mathbf{c}} := \left\{ 
\tau \in \mathcal{D}(\mathsf{R}_1 \otimes \cdots \otimes \mathsf{R}_{m+1}) 
\;\middle|\;
\mathcal{E}_{\mathsf{R}_{m+1}, \mathsf{R}_1 \otimes \cdots \otimes \mathsf{R}_m}(\tau) \leq c_{m+1}
\right\},
\]
where $\tau \in \{\rho, \sigma\}$, $\mathbf{c} \in \{\mathbf{b}, \mathbf{d}\}$, and $\mathsf{R} \in \{\mathsf{A}, \mathsf{B}\}$ depending on the prover. We set $\mathcal{T}^{(0)}_{\tau,\mathbf{c}} := \mathcal{D}(\mathsf{R}_1)$.

\medskip

If $b_i = b$ and $d_i = d$ for all $i$, then~\cref{eq:4level} and~\cref{eq:general_i} define equivalent optimization problems. Since the feasible regions are convex and compact, and the objective function is bilinear, the optimal value can be computed by a $\mathsf{BEQ\Sigma}_4$ verifier. Consequently, for all $i > 4$, we obtain $\mathsf{BEQ\Sigma}_i^{\{b\},\{d\}} = \mathsf{BEQ\Sigma}_4^{\{b\},\{d\}}$, where the notation $\{b\}, \{d\}$ indicates that $b_i = b$ and $d_i = d$ for all $i$.
\end{proof}

\begin{corollary}[Equivalence of quantum interactive proof classes at the second level]
\label{cor:hierarchy-equality} $\mathsf{QEPH} = \mathsf{QE\Sigma}_2 = \mathsf{BEQ\Sigma}_2 = \mathsf{Q\Sigma}_2 = \mathsf{QRG}(1)$.
\end{corollary}

\begin{proof}
The equality $\mathsf{QEPH} = \mathsf{QE\Sigma}_2$ follows from~\cite{grewaljustinqeph}. By definition, $\mathsf{BEQ\Sigma}_2 = \mathsf{QE\Sigma}_2 = \mathsf{Q\Sigma}_2$. The equivalence with $\mathsf{QRG}(1)$ follows from the known correspondence between constant-round quantum refereed games and second-level quantum interactive proofs~\cite{qrg1}.
\end{proof}

\begin{corollary}[Collapse of separable quantum polynomial hierarchy]
\label{cor:separable-classes}
The class $\mathsf{SepQPH}$ collapses to a two-round convex optimization problem. Moreover, for all $i > 4$, it holds that $\mathsf{SepQ\Sigma}_i = \mathsf{SepQ\Sigma}_4$.
\end{corollary}

\begin{proof}
This follows directly by setting all entanglement bounds to zero in~\cref{thm:beqph-collapse}, thereby enforcing separability at each round, in accordance with~\cref{def:ree}.
\end{proof}

\begin{corollary}[Separable hierarchies do not simulate the quantum polynomial hierarchy]
\label{cor:sepqph_cant_simulate_qph}
$\mathsf{SepQPH}$ cannot simulate $\mathsf{QPH}$.
\end{corollary}

\begin{proof}
Setting entanglement bounds to zero ensures that all states are separable, forming a convex set. However, the optimization problem for $\mathsf{QPH}$ is generally non-convex, as the set of product states is non-convex. Hence, $\mathsf{SepQPH}$ cannot simulate the full power of $\mathsf{QPH}$.
\end{proof}

\subsection{Structural comparison of hierarchy collapses}

\cref{cor:hierarchy-equality} establishes that the $\mathsf{BEQPH}$ hierarchy can be collapsed to corresponding two-round convex optimization problems, analogous to the collapse of $\mathsf{QEPH}$ to its second level. This demonstrates that constraining the \emph{quality} of relative entropy of entanglement at each round still preserves the tractability of the quantum hierarchy, including the zero-entanglement case. At zero entanglement, we obtain the class $\mathsf{SepQPH}$, whose associated optimization problem is also convex due to the convexity and compactness of the set of separable quantum states. In $\mathsf{QEPH}$, the collapse mechanism relies crucially on the ability of provers to share entanglement across rounds~\cite{grewaljustinqeph}. The structure of games in that setting allows for quantum correlations that effectively ``short-circuit'' higher-level quantifier alternations, collapsing the hierarchy to the second level. Since the second level is expressible as a convex optimization problem, this yields tractability.

In contrast, for the bounded-entanglement version of $\mathsf{QEPH}$, denoted by $\mathsf{BEQPH}$, we do not observe a full collapse to the second level. However, we still achieve tractability at every level: each level of $\mathsf{BEQPH}$ can be effectively reduced to a corresponding two-round max-min optimization problem over convex sets. Moreover, we show that for all $i > 4$, the class $\mathsf{BEQ\Sigma}_i$ has the same formulation as $\mathsf{BEQ\Sigma}_4$, thereby implying an unconditional collapse of the hierarchy to the fourth level. Furthermore, as established in~\cref{cor:separable-classes}, setting the entanglement bounds $\mathbf{b}, \mathbf{d}$ to zero yields the hierarchy $\mathsf{SepQPH}$, in which each of the sets $\mathcal{T}_{\tau,\mathbf{c}}^{(m)}$ consists of separable states over $\mathsf{R}_1 \otimes \cdots \otimes \mathsf{R}_{m+1}$. As shown in~\cref{cor:sepqph_cant_simulate_qph}, this restriction prevents $\mathsf{SepQPH}$ from simulating the full power of $\mathsf{QPH}$, whose optimization problems are generally non-convex. However, if we further restrict the sets $\mathcal{T}_{\tau,\mathbf{c}}^{(m)}$ to contain only pure product states, we can in fact simulate $\mathsf{QPH}$—a fact that is essentially \emph{trivial} to verify.

\paragraph{Implications for quantum complexity theory.}
Our results offer a refined perspective on the structure of the quantum entangled polynomial hierarchy. While the unconditional collapse of $\mathsf{QEPH}$ to its second level has been attributed to the unbounded entanglement permitted between rounds, we show that introducing upper bounds on the allowed relative entropy of entanglement at each round still yields a form of structural simplification. In particular, the bounded-entanglement quantum polynomial hierarchy retains tractability at every level: although it does not collapse completely to the second level, each level reduces to a two-round convex optimization problem.

When the entanglement bounds are set to zero, we obtain the separable quantum polynomial hierarchy, where all intermediate states are separable across rounds. In this regime, the computational simplicity of $\mathsf{BEQPH}$ persists, though $\mathsf{SepQPH}$ is strictly more computationally expressive than $\mathsf{QEPH}$, as it does not collapse unconditionally to its second level. Nonetheless, both $\mathsf{SepQ\Sigma_2}$ and $\mathsf{SepQ\Sigma_3}$ remain tractable, as the corresponding optimization problems are convex and hence efficiently solvable via semidefinite programming.


\section{Open problems and future directions}
Our results naturally lead to several intriguing open questions that we believe merit further investigation. We organize these into two main categories corresponding to our primary contributions.

\subsection{Uniqueness in quantum-classical probabilistically checkable proofs}
The equivalence $\mathsf{UniqueQCPCP} = \mathsf{QCPCP}$ established in~\cref{sec:uqcpcp_eq_qcpcp} opens several natural avenues for future research, particularly concerning the robustness of this collapse and its extensions to broader quantum proof systems.

\paragraph{Relation between \(\mathsf{QPCP}\) and \(\mathsf{UniqueQPCP}\).} 
A natural question arising from our work is whether similar results hold in the quantum proof setting as well. 
Currently, the only known result is a quantum oracle separation between \(\mathsf{UniqueQMA}\) and \(\mathsf{QMA}\), shown in~\cite{anshu2024uniqueqmavsqmaoracle}, 
which itself appears to be a highly nontrivial and challenging problem.

\paragraph{Stronger reductions for the $\mathsf{UniqueQCPCP} = \mathsf{QCPCP}$ equivalence.}
    A natural direction is to investigate whether this equivalence remains true under stronger or more restricted reduction models. Our proof relies on $\mathsf{BQ}$-operator and randomized reductions, but it would be interesting to establish the equivalence under $\mathsf{BQP}$-truth-table reductions (i.e., non-adaptive quantum polynomial-time reductions), or even deterministic polynomial-time reductions. Establishing this would demonstrate that the collapse to $\mathsf{UniqueQCPCP}$ holds under more conservative assumptions about how the reduction accesses the verifier. Conversely, proving a separation would highlight fundamental structural distinctions between the two systems that are obscured by more powerful reduction models.

\paragraph{Extension to multi-prover and entangled variants.}
    Another compelling direction is to investigate whether the equivalence between $\mathsf{UniqueQCPCP}$ and $\mathsf{QCPCP}$ extends to broader classes of quantum proof systems. For instance, in the multi-prover setting, does a similar collapse hold between uniqueness and generality for $\mathsf{MIP}^*$-style variants of $\mathsf{QCPCP}$, where the verifier interacts with multiple (possibly entangled) provers? In such settings, uniqueness constraints may behave quite differently due to entanglement and non-local correlations, making it unclear whether the collapse still holds. Moreover, the notion of ``uniqueness'' itself becomes more subtle when multiple provers are involved. Exploring these extensions could yield new insights into the role of entanglement and prover non-locality in quantum PCP constructions.

\paragraph{Parallel repetition for $\mathsf{UniqueQCPCP}$.}
    In classical PCP theory, parallel repetition theorems are crucial for amplifying soundness while preserving completeness~\cite{Raz1998}. A natural question is whether there exists a quantum analogue of the parallel repetition theorem for $\mathsf{UniqueQCPCP}$ systems. Specifically, if a quantum verifier checks multiple independent copies of a $\mathsf{UniqueQCPCP}$ proof in parallel, does the soundness error decrease exponentially with the number of repetitions while completeness remains unaffected? Such a result would significantly amplify the reliability of $\mathsf{UniqueQCPCP}$ verifiers and could potentially be used to derive stronger quantum hardness-of-approximation results. However, the quantum setting presents unique challenges: entangled cheating strategies across multiple proof copies and the delicate structure of quantum measurements may disturb classical intuitions, making the existence of such a theorem highly non-trivial and potentially impactful.

\subsection{Structure of the quantum polynomial hierarchy}
Our investigation of $\mathsf{QPH}$ and its entanglement-bounded variant $\mathsf{BEQPH}$, presented in~\cref{sec:qph}, reveals a rich structure that invites further exploration. Several fundamental questions remain open regarding the relationships between these hierarchies and other quantum complexity classes.

\paragraph{Towards a uniform Karp–Lipton theorem for $\mathsf{QPH}$.}
    One of the most significant open questions in quantum complexity theory is whether there exists a uniform (advice-free) Karp–Lipton collapse theorem for $\mathsf{QPH}$. Our non-uniform result shows that $\mathsf{QMA} \subseteq \mathsf{BQP}/\mathsf{qpoly}$ implies $\mathsf{QPH} \subseteq \mathsf{Q\Sigma}_2/\mathsf{qpoly}$, but the uniform case remains elusive. A uniform collapse would require showing that $\mathsf{QMA} \subseteq \mathsf{BQP}$ implies $\mathsf{QPH} \subseteq \mathsf{Q\Sigma}_2$. Such a result currently seems out of reach, as it would likely imply $\mathsf{QMA}(2) \subseteq \mathsf{PSPACE}$ proving which appears beyond the scope of existing techniques~\cite{qph_original}.

\paragraph{Relating $\mathsf{QPH}$ to other quantum complexity classes.}
    A fundamental research direction is to better understand how various levels of $\mathsf{QPH}$ relate to other well-studied quantum complexity classes such as $\mathsf{QRG}(2)$, $\mathsf{PSPACE}$, $\mathsf{QMA}(2)$, $\mathsf{QMA}$, and $\mathsf{QCMA}$. Of particular interest is whether any level of $\mathsf{QPH}$ captures the full computational power of $\mathsf{PSPACE}$. While we know that $\mathsf{QPH} \subseteq \mathsf{PSPACE}$ by definition, the reverse inclusion remains a tantalizing open problem that would provide a complete characterization of quantum hierarchical proof systems.


\paragraph{Formal containment of $\mathsf{BEQPH}$ in $\mathsf{QPH}$.}
    Another central open problem is to determine whether the class $\mathsf{BEQPH}$, defined by its consistent proof structure and bounded entanglement, is strictly contained within the general quantum polynomial hierarchy. We conjecture that $\mathsf{BEQPH} \subsetneq \mathsf{QPH}$ and $\mathsf{SepQPH} \subsetneq \mathsf{QPH}$, as the structural constraints and the collapse to the fourth level observed in $\mathsf{BEQPH}$ suggest a significant reduction in expressive power. However, establishing such a separation formally remains challenging, due to the inherent difficulty of characterizing the full complexity of $\mathsf{QPH}$.

\paragraph{Comparison with oracular quantum hierarchies.}
    An intriguing direction for future work is to compare the bounded-entanglement quantum polynomial hierarchy  with oracular quantum hierarchies, such as the iterated quantum Merlin-Arthur hierarchy $\mathsf{QMAH} := \mathsf{QMA}^{\mathsf{QMA}^{\mathsf{QMA}^{\cdots}}}$, where the $\mathsf{QMA}$ oracle is applied polynomially many times. It remains an open question whether the structural consistency constraints and entanglement bounds imposed in $\mathsf{BEQPH}$ make it as computationally powerful as $\mathsf{QMAH}$, or whether the iterated access to $\mathsf{QMA}$ oracles provides strictly greater computational power. As noted in~\cite{grewaljustinqeph}, resolving this question could shed light on the internal structure of $\mathsf{QRG}(1)$. Since our work establishes $\mathsf{BEQPH} \subseteq \mathsf{QRG}(1)$ and it is known that $\mathsf{QMAH} \subseteq \mathsf{CH}$~\cite{Vinkhuijzen2018}, a proof of the containment $\mathsf{BEQPH} \subseteq \mathsf{QMAH}$ would imply that $\mathsf{QRG}(1) \subseteq \mathsf{CH}$ as well, thereby placing tighter constraints on the computational power of constant-round quantum refereed games.

\section*{Acknowledgments}
\addcontentsline{toc}{section}{Acknowledgments}
J.L. acknowledges useful discussions with Nhat A. Nghiem. The authors are grateful to Dorian Rudolph for his initial feedback on this work. This work was supported by the National Research Foundation of Korea (NRF) through a grant funded by the Ministry of Science and ICT (Grant No. RS-2025-00515537). This work was also supported by the Institute for Information \& Communications Technology Promotion (IITP) grant funded by the Korean government (MSIP) (Grant Nos. RS-2019-II190003 and RS-2025-02304540), the National Research Council of Science \& Technology (NST) (Grant No. GTL25011-401), and the Korea Institute of Science and Technology Information (KISTI) (Grant No. P25026). 

\addcontentsline{toc}{section}{References}
\bibliographystyle{alpha}
\bibliography{citation.bib}

\end{document}